\documentclass[conference]{IEEEtran}
\IEEEoverridecommandlockouts
\usepackage{cite}
\usepackage{amsmath,amssymb,amsfonts}
\usepackage{algorithmic}
\usepackage{algorithm}
\usepackage{graphicx}
\usepackage{textcomp}
\usepackage{xcolor}

\usepackage{cleveref}

\usepackage{booktabs}
\def\BibTeX{{\rm B\kern-.05em{\sc i\kern-.025em b}\kern-.08em
    T\kern-.1667em\lower.7ex\hbox{E}\kern-.125emX}}
\usepackage{amsthm}

\newtheorem{theorem}{Theorem}
\IEEEoverridecommandlockouts
\begin{document}

\title{Employing the Structural Power to Achieve
Supply-Demand Balanced Payment Channel Networks}

\author{
\IEEEauthorblockN{
Shuyao Xiao\textsuperscript{1},
Shengling Wang\textsuperscript{1,*},
Hongwei Shi\textsuperscript{1},
Weicheng Wang\textsuperscript{2},
Anlin Chen\textsuperscript{1}
}
\IEEEauthorblockA{
\textsuperscript{1}College of Artificial Intelligence, Beijing Normal University, Beijing, China\\
\textsuperscript{2}Beijing Institute of Technology, Zhuhai, China\\
xiaoshuyao@mail.bnu.edu.cn, wangshengling@bnu.edu.cn, hongweishi@mail.bnu.edu.cn
}
}

\maketitle

\begingroup
\renewcommand{\thefootnote}{*}
\footnotetext{Corresponding author.}
\endgroup

\begin{abstract}
Blockchain technology faces scalability challenges because transactions must be validated and recorded across the network. Payment channel networks (PCNs) improve efficiency by moving transactions off-chain and recording only critical interactions on the mainnet. However, PCNs require pre-deposited channel balances (supply) to match transaction demands (demand), and insufficient balances cause supply shortages. Existing approaches, including transaction path optimization and channel balance allocation, address this problem but incur high costs due to dynamic adaptation to fluctuating demands.
We reveal a correlation between balance deficits in PCNs and network topology by uncovering the structural organization of liquidity allocation. This enables supply–demand conflicts to be mitigated from a static topological perspective through channel reconfiguration, without additional balance replenishment. We further introduce payment topological entropy (PTE), an information-theoretic metric that quantifies each node’s deviation from the global average connection pattern and captures structural properties of balance supply. Based on PTE, we design MaxPTE, a topology optimization algorithm that reorganizes balance allocation across channels through structural reconfiguration, aligning static balance distribution with dynamic transaction demand.
Extensive experiments show that MaxPTE reduces balance deficits by 27.25\%, increases average maximum flow by 12.23\%, and decreases transaction failure probability by 25.96\%, outperforming existing benchmarks. The method also remains robust across diverse balance-demand distributions, improving supply–demand balance without prior demand prediction.
\end{abstract}


\section{Introduction}
Blockchain technology ensures decentralization and transaction security. However, it faces significant scalability challenges, because every transaction must be validated by all nodes and recorded on the blockchain ledger, leading to low throughput and high confirmation latency as the number of transactions increases. To address this, payment channel networks (PCNs) have emerged as a solution by enabling off-chain transactions that reduce reliance on the blockchain mainnet, thereby improving transaction efficiency. Specifically, PCNs allow any two blockchain nodes to pre-establish a dedicated off-chain payment channel using smart contracts ~\cite{poon2016bitcoin}. In contrast to executing each transfer directly on the blockchain mainnet, this channel accommodates high-frequency transactions and submits only critical information to the mainnet upon channel creation, final settlement, or disputes. Consequently, this mechanism significantly reduces on-chain submissions and alleviates the load on the blockchain mainnet.

Although payment channels support bidirectional transactions, the
available balance for each direction is determined by the current
channel state. Since our objective is to characterize how pre-deposited
liquidity is allocated among outgoing payment directions at a given
state, we model the PCN as a directed weighted graph, where each directed
edge represents the transferable balance from one node to another at that
state.

To operate effectively, PCNs require each off-chain channel to maintain a specific pre-deposited balance. This balance ensures collateralization and settlement for off-chain transfers. It offers sufficient balance for direct payments between nodes, prevents overspending, and ensures the necessary balance for relay transactions using the channel. The balance in each channel is independently locked, cannot flow across channels, and cannot be dynamically reallocated by nodes to replenish other channels; it can only be configured ex ante according to the smart contract~\cite{poon2016bitcoin}. This principle of independent locking ensures a one-to-one correspondence between balance and channel state, securing the safety, verifiability, and tamper-resistance of off-chain payments at the mechanism level. Once the payment channel is established, the balance remains fixed and cannot be adjusted or redistributed.

From a supply–demand perspective, this smart contract-based static locking mechanism defines the pre-deposited balance in channels as supply, while direct and relay payment requests of this channel constitute demand. Since balance cannot move across channels or be expanded ex post~\cite{khalil2017revive}, the total supply in a channel is determined solely by the pre-deposited amount set at channel establishment. When transaction demand exceeds the locked balance in a channel, transactions may fail even if other channels have enough balance, as cross-channel reallocation is infeasible. The structural supply–demand mismatch caused by static balance locking is a fundamental cause of insufficient pre-deposited balance problems in PCNs.

Existing solutions to the issue of insufficient supply in PCNs can be generally classified into two categories: \emph{the transaction path optimization} \cite{sivaraman2020high,lin2020fstr,bagaria2020boomerang,rahimpour2021spear,zhang2021robustpay+,liu2024balanced} and 
\emph{the channel balance planning} \cite{khalil2017revive,avarikioti2022hide,hong2022cycle,ge2022shaduf,li2020secure,ni2023utility,wang2023fence}. 
The former strategically selects the optimal transaction path from the available channels to reduce balance shortage, while the latter reasonably allocates the pre-deposited balance on the channel to achieve the supply-demand equilibrium. 
However, these methods mainly address the supply-demand imbalance through demand adaptation or balance redistribution, which may incur additional adjustment overhead under dynamic environments. Although intuitive,  they are costly and temporary since supply optimization must dynamically respond to demand variations. It is vulnerable to the high-cost issue of on-demand adaptation in complex environments.

De facto, the real solution often lies beyond superficial manifestations. This indicates that only by recapitulating the fundamental factors of the phenomenon can we embrace the dominant force that functions at a fundamental level. 
Aware of this, 
our first question is: \emph{ What is the dominant force behind the supply-demand dilemma that is not demand-awareness but can still systematically alleviate the balance deficit?}\looseness=-1

To answer the first question, we explore solutions from the topological structure of PCNs, which can be adjusted through channel reconfiguration mechanisms. We conducted empirical experiments in the Lightning Network Daemon \cite{LightningNetworkDaemon} to explore the relationship between the topologies of PCNs and balance deficits. By balance deficits, we mean the difference between the demand and the deposited balance within a channel.
\begin{figure}[t]
    \centering
    \includegraphics[width=1\linewidth]{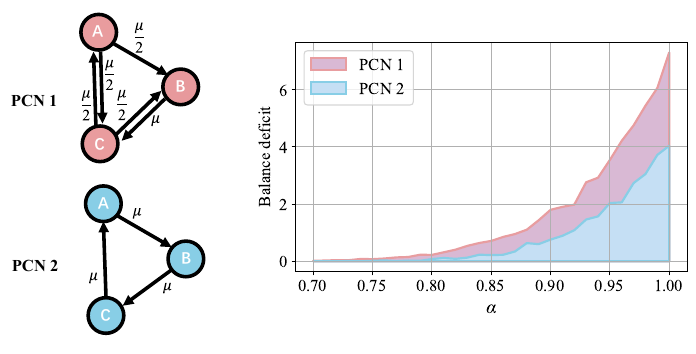}
    \caption{Two PCNs with 3 nodes and their balance deficits. Assume that each node possesses the equal balance capital $\mu$. These two networks show different topologies formed by the allocation of their pre-deposited balance at this moment, as shown on the left. Let the amount of transaction balance demand be proportional to the balance with parameter $\alpha\in[0.7,1]$, and it follows  Poisson distribution \cite{ross2014introduction} with mean $\lambda=\alpha\times\mu$. The figure on the right reflects the comparison of balance deficits between PCN 1 and PCN 2.}
    \label{2pcnCompare}
    \vspace{-0.2cm}
\end{figure}

\begin{figure}[t]
    \centering
    \includegraphics[width=0.95\linewidth]{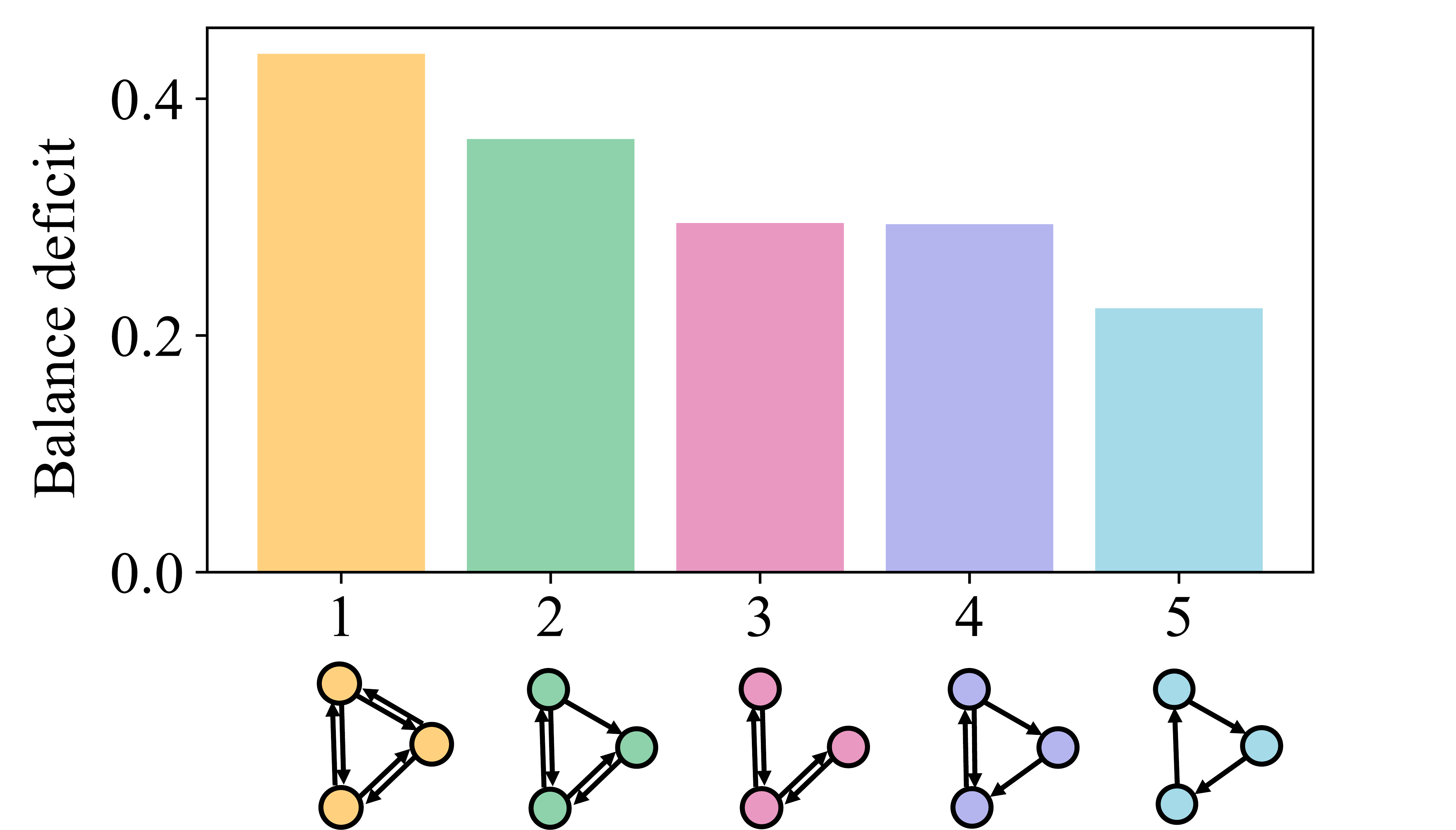}
    \caption{Balance deficits among different topologies.}
    \label{different_topologies}
    \vspace{-0.2cm}
\end{figure}

We tested the expected balance deficits generated in each channel within two different PCN topological structures, as shown in the left panel of \cref{2pcnCompare}. 
The experiment \footnote{The blockchain protocol of PCNs (e.g., Bitcoin's script) requires both parties to jointly authorize channel state updates. Since the transferable balance in each direction is determined by the current channel state, we represent the instantaneous liquidity state as directional balance flows and use directed graphs to capture the allocation of transferable balances in PCNs.} was conducted 100 times, and the results were averaged. The right figure in \cref{2pcnCompare} shows that, with the same demand, the total balance deficit of PCN 1 significantly exceeds that of PCN 2. 

To further verify the impact of network topology on balance deficits in PCNs, we extended the experimental setup to test additional connection patterns in three-node PCNs. We strictly ensured that each node possessed at least one outgoing channel and one incoming channel so that all nodes could effectively participate in the circulation of balances within the network. As shown in \cref{different_topologies}, PCNs with different topological structures exhibit significant differences in balance deficits. These preliminary findings suggest that PCN topology is an important factor affecting balance shortages in the network~\cite{seres2020topological}. Building upon this key observation, we raise the second core research question: \emph{What topological characteristics of a PCN can effectively reduce balance deficits?}

To address the above-mentioned question, we adopt an information-theoretic perspective and introduce the concept of \emph{payment topological entropy} ($PTE$), which quantifies the deviation of each node from the global average connection pattern. It characterizes the implicit regularities of the PCN’s balance supply system, which transforms the topological features of a PCN into a quantifiable information entropy parameter. This not only explains the underlying causes of the differences in balance deficits observed between various topologies, such as PCN 1 and PCN 2, but also verifies through further experiments that PCNs with lower balance deficits indeed exhibit higher $PTE$.
Building on this foundation, the third question we aim to answer is: \emph{Through what mechanism can we construct a PCN topology that minimizes the risk of balance deficits based on this metric?}

To tackle the third question, we leverage the characteristic that a PCN topology can be adjusted at low cost via modifying smart contracts. We propose a PTE-guided topology optimization algorithm that iteratively
searches for channel consolidations improving the structural organization
of liquidity. The algorithm iteratively merges redundant channels between nodes, efficiently aggregating balances that were previously dispersed across multiple channels. While optimizing node connection patterns to increase $PTE$, it also maintains strong network connectivity through predefined constraints. Our algorithm effectively reduces the risk of channel balance deficits due to supply-demand imbalances while maintaining transaction efficiency, achieving coordinated optimization of the PCN topological order and transaction performance.

Our work moves beyond traditional supply-demand management by
uncovering the structural factors that influence balance deficits in
PCNs. We propose a demand-robust solution that mitigates the complexity of performance optimization caused by the unpredictability of transient demands and the uncertainty of unknown demands. Our contributions are both pragmatic and theoretical, and can be summarized as follows:
\begin{enumerate}
    \item We reveal the relationship between balance deficits and PCN topologies from the perspective of structural liquidity organization.
    Therefore, this insight provides a new perspective for addressing dynamic supply-demand imbalance. Such an insight enables a cost-efficient approach, whereby balance deficits can be mitigated without increasing the total balance through topology-level channel reconfiguration.
    \item We introduce $PTE$, a metric that characterizes the structural properties of balance supply in PCNs. Building upon this metric, we further design a PCN topology optimization algorithm. Under the balance-locking rule, this mechanism efficiently aligns the static balance distribution with the dynamic transaction demand, effectively reducing balance deficits from a perspective of PCNs' topological structure.
    \item Extensive simulations and performance evaluations validate the effectiveness of the proposed approach. Extensive experiments show that our MaxPTE algorithm significantly reduces balance deficits by 27.25\%, increases average maximum flow by 12.23\%, and decreases transaction failure probability by 25.96\%, which surpasses other benchmarks. Moreover, our method demonstrates considerable robustness across various balance-demand distributions, demonstrating its effectiveness in improving the supply–demand balance.
\end{enumerate}

\section{Payment Topological Entropy}
\label{sec: Payment Topological Entropy}

\subsection{Definition}

To quantify the information characteristics of PCN topology, inspired by complex network analysis~\cite{newman2003structure}, we introduce payment topological entropy (PTE) as a topology-level indicator of structural liquidity organization,  which measures the deviation of each node's connection pattern from the global average and thereby captures the structural regularities of the balance supply system in PCNs.

\textbf{Definition 1 (Payment topological entropy, $PTE$).}
The $PTE$ of a PCN topology is defined as the average structural deviation of all nodes in the network:

\begin{equation}
PTE=
\frac{1}{n}
\sum_{i=1}^{n}
TVD(W_i^{out},W^{in}),
\label{PTE}
\end{equation}

where $n$ denotes the total number of nodes in the PCN, and
$TVD(W_i^{out},W^{in})$ represents the deviation of node $i$'s outgoing balance allocation pattern from the network-level incoming balance distribution, measured by the Total Variation Distance (TVD)~\cite{cover2012elements}.

Specifically,

\begin{equation}
W_i^{out}=(w_{i1},w_{i2},\ldots,w_{in})
\end{equation}

denotes the normalized outgoing balance distribution of node $i$, where

\begin{equation}
w_{ij}
=
\frac{c_{ij}}
{\sum_{k\in i^{out}}c_{ik}},
\end{equation}

where $c_{ij}$ denotes the balance allocated from node $i$ to the channel established with node $j$, and $\sum_{k\in i^{out}}c_{ik}$ represents the total balance reserved by node $i$. Therefore, $w_{ij}$ indicates the proportion of node $i$'s total reserved balance allocated to node $j$.

In contrast, $W^{in}$ describes the network-level average incoming balance distribution, which is calculated by aggregating the normalized outgoing balance distributions of all nodes:

\begin{equation}
W^{in}
=
(w^{in}_{1},w^{in}_{2},\ldots,w^{in}_{n}),
\end{equation}

where each component is defined as

\begin{equation}
w^{in}_{j}
=
\frac{1}{n}
\sum_{i=1}^{n}w_{ij}.
\end{equation}

Here, $w^{in}_{j}$ represents the average proportion of normalized outgoing balance received by node $j$ from all nodes in the network.
Accordingly, the Total Variation Distance between $W_i^{out}$ and $W^{in}$ is calculated as

\begin{equation}
TVD(W_i^{out},W^{in})
=
\frac{1}{2}
\sum_{j=1}^{n}
\left|
w_{ij}-w^{in}_{j}
\right|.
\label{TV}
\end{equation}

The value of $TVD(W_i^{out},W^{in})$ reflects the structural difference between node $i$'s local liquidity allocation and the global balance circulation pattern. A larger value indicates that node $i$ has a more distinctive balance allocation pattern compared to the network average, while a smaller value suggests that node $i$ follows a distribution closer to the overall network pattern.

\subsection{Structural Interpretation of $PTE$.}

In a PCN, available liquidity is determined not only by the total amount of pre-deposited balance, but also by how that balance is structurally distributed across payment channels. Due to the channel-level balance-locking mechanism, balances assigned to different outgoing channels form independent liquidity pools and cannot be freely reallocated after channel establishment. Consequently, two PCNs with the same total balance may exhibit different balance deficits because their liquidity is organized through different topological structures.

For node $i$, $W_i^{out}$ describes its local liquidity allocation pattern, representing the proportions of its reserved balance assigned to different outgoing channels. In contrast, $W^{in}$ describes the network-level balance reception pattern obtained by aggregating the normalized outgoing allocations of all nodes. Comparing these two distributions characterizes how the liquidity role of an individual node differs from the overall balance circulation pattern of the network.

When $W_i^{out}$ is close to $W^{in}$, node $i$ follows a liquidity allocation pattern similar to the network-wide structure. If most nodes exhibit similar allocation behaviors, the network tends to have homogeneous liquidity roles, with multiple nodes providing overlapping liquidity functions. Under the balance-locking mechanism, these repeated allocation patterns fragment available liquidity into multiple independently constrained pools, reducing the efficiency of balance circulation.

In contrast, a larger $TVD(W_i^{out},W^{in})$ indicates that node $i$ exhibits a more differentiated liquidity role. Its balance allocation pattern provides a distinct contribution to the overall circulation structure rather than replicating existing allocation patterns. At the network level, differentiated liquidity roles reduce redundant balance allocation and form more effective balance-flow corridors, enabling the existing liquidity supply to support larger transferable flows.

Therefore, $PTE$ measures the structural differentiation of liquidity roles among nodes. A low $PTE$ indicates similar balance allocation patterns and the network contains redundant liquidity structures, while a high $PTE$ suggests complementary liquidity roles and a more organized balance supply structure. The effectiveness of $PTE$ stems from the balance-locking mechanism of PCNs. When multiple nodes adopt similar liquidity allocation patterns, the available balance is repeatedly partitioned into independent pools, causing structural redundancy and limiting balance circulation. Therefore, a higher $PTE$ indicates not merely greater topological diversity, but a more complementary organization of liquidity roles that reduces redundant balance allocation.

\begin{figure}[t]
    \centering
    \includegraphics[width=1\linewidth]{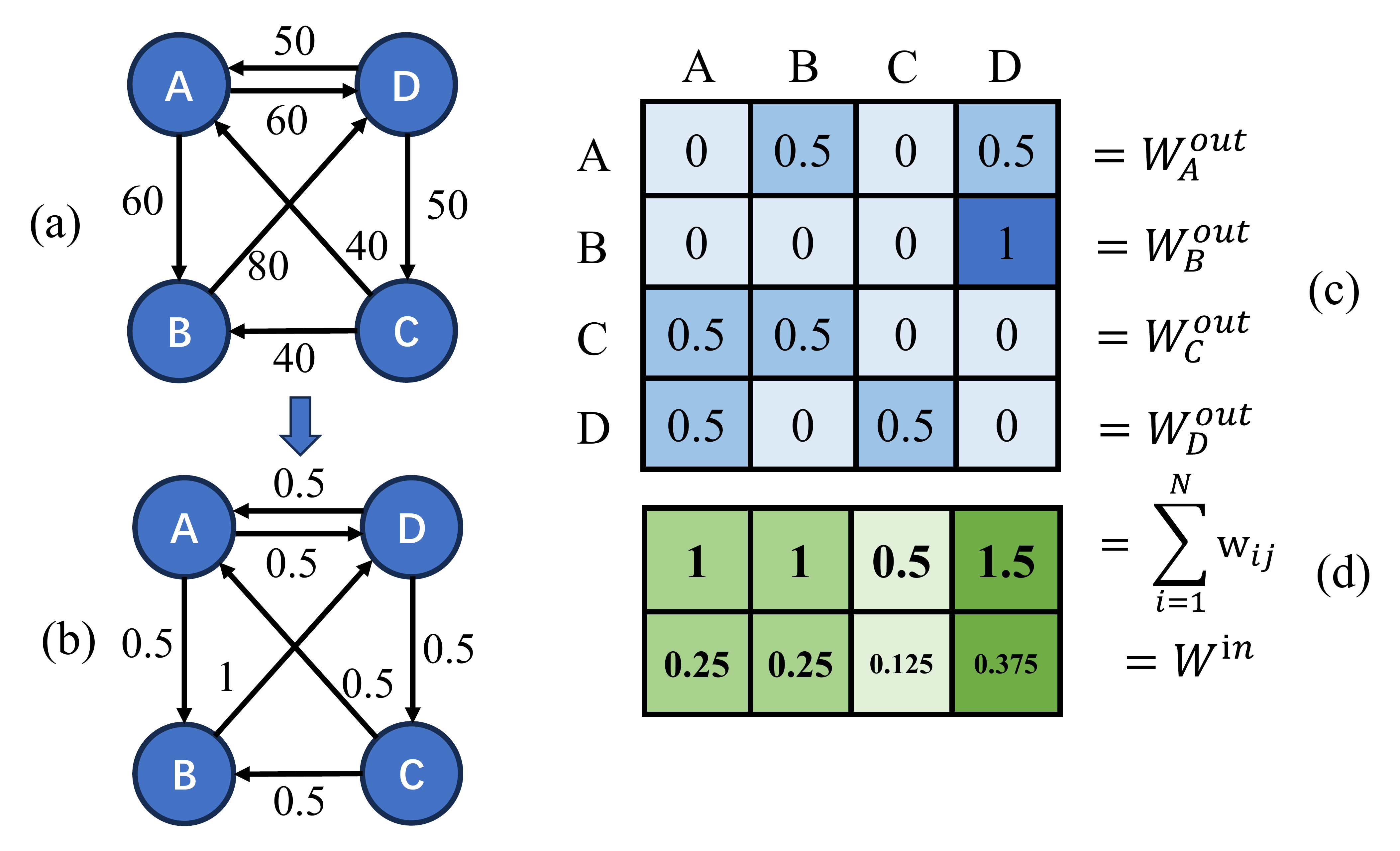}
    \caption{Example computation of $PTE$ on a toy PCN and its normalized matrices.}
    \label{pte example}
    \vspace{-0.2cm}
\end{figure}

\cref{pte example} illustrates an example of calculating the $PTE$ of a network. 
\cref{pte example}(a) shows balances reserved in all channels in the network. 
\cref{pte example}(b) represents the result of normalizing the pre-deposited funds of each node. 
For instance, node $A$ allocates 60 coins to the channel from $A$ to $B$ and another 60 coins to the channel from $A$ to $D$. 
After normalization, the weight of each outgoing edge from node $A$ becomes $60 / (60 + 60) = 0.5$.
\cref{pte example}(c) represents the normalized transition matrix of the PCN. 
Each row describes the proportion distribution of the outgoing balances of the corresponding node to each connected node, i.e., $W_i^{out}$.
\cref{pte example}(d) shows the column-wise summation of all entries in \cref{pte example}(c), corresponding to the ratio of total incoming balances received by each node. 
The first row lists $\sum_{i=1}^{n} w_{ij}$, and the second row shows the normalized results, i.e., $W^{in}$.

Substituting these values into \eqref{PTE} and \eqref{TV}, we obtain
\begin{align*}
PTE 
&= \frac{1}{n} \sum_{i=\left\{A,B,C,D\right\}} 
\mathrm{TVD}\!\left(W_i^{\mathrm{out}},\, W^{\mathrm{in}}\right)\\
&= \frac{1}{4}\sum_{i=\left\{A,B,C,D\right\}} \left[\frac{1}{2}\sum_{j=\left\{A,B,C,D\right\}}
\bigl|w_{ij}-W^{\mathrm{in}}\bigr|\right] \\
&= \frac{1}{4} \times \frac{1}{2} \Big[
\bigl|(0, 0.5, 0, 0.5) - (0.25, 0.25, 0.125, 0.375)\bigr| \\
&\quad\quad\quad\quad + \bigl|(0, 0, 0, 1) - (0.25, 0.25, 0.125, 0.375)\bigr| \\
&\quad\quad\quad\quad + \bigl|(0.5, 0.5, 0, 0) - (0.25, 0.25, 0.125, 0.375)\bigr| \\
&\quad\quad\quad\quad + \bigl|(0.5, 0, 0.5, 0) - (0.25, 0.25, 0.125, 0.375)\bigr|
\Big] \\
&= 0.53125.
\end{align*}

\begin{table}[t]
  \caption{$PTE$ and balance deficit of different PCNs.}
  \label{tab:pte-deficit}
  \centering
  \begin{tabular}{cccccc}
    \toprule
    & PCN 1 & PCN 2 & PCN 3 & PCN 4 & PCN 5 \\
    \midrule
    Topology &
      \includegraphics[height=8mm]{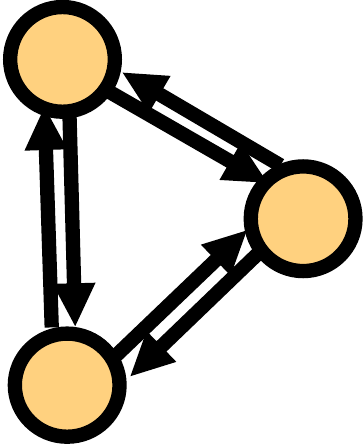} &
      \includegraphics[height=8mm]{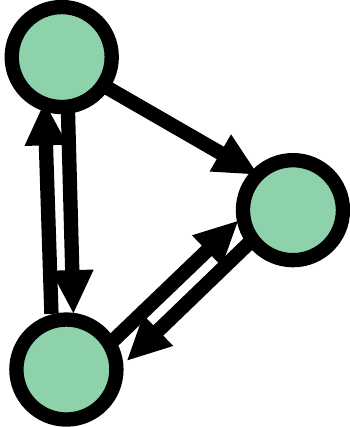} &
      \includegraphics[height=8mm]{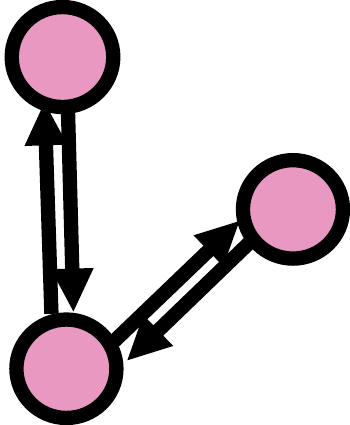} &
      \includegraphics[height=8mm]{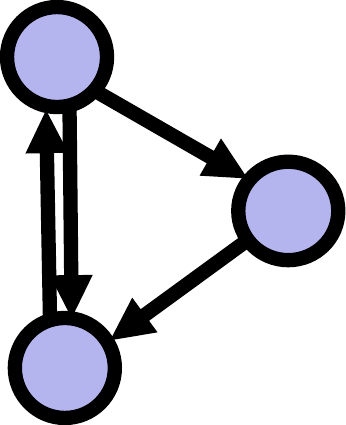} &
      \includegraphics[height=8mm]{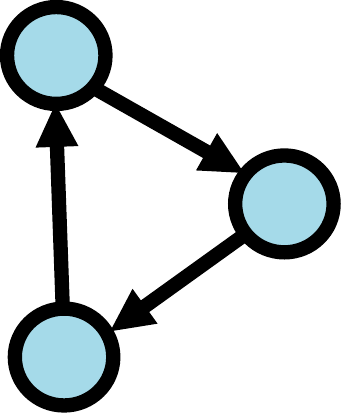} \\
    $PTE$ & 0.33 & 0.39 & 0.44 & 0.50 & 0.67 \\
    Balance deficit & 0.438 & 0.366 & 0.295 & 0.295 & 0.223 \\
    \bottomrule
  \end{tabular}
\end{table}
\subsection{Connection between $PTE$ and Topology Reconfiguration}
\label{subsec:pte_reconfiguration}

The preceding analysis explains how $PTE$ characterizes structural liquidity organization, while topology reconfiguration determines how this property can be improved. Consider a node whose reserved balance is distributed across multiple outgoing channels. When several channels provide similar liquidity functions, maintaining them separately disperses the available balance into multiple independent pools without proportionally improving balance circulation capability. Consolidating such channels preserves the total balance while reorganizing the liquidity allocation toward more distinctive payment directions.

Formally, suppose that the balances of two outgoing channels $(i,p)$ and $(i,q)$ are consolidated into channel $(i,p)$. The operation is defined as
\begin{equation}
c'_{ip}=c_{ip}+c_{iq}, \qquad c'_{iq}=0,
\label{eq:channel_merge}
\end{equation}
while the total reserved balance of node $i$ remains unchanged:
\begin{equation}
\sum_{j=1}^{n}c'_{ij}=\sum_{j=1}^{n}c_{ij}.
\label{eq:balance_conservation}
\end{equation}

This operation changes the normalized outgoing distribution $W_i^{out}$ while preserving the total liquidity supply. Among feasible channel pairs, consolidations that increase $PTE$ correspond to reorganizing redundant liquidity allocation into more differentiated structural roles. Therefore, $PTE$ provides a topology-level criterion for identifying which channel configurations should be consolidated without requiring prior knowledge of transaction demands.

However, topology optimization should balance structural differentiation and network accessibility. Excessive channel consolidation may reduce available payment paths and affect transaction efficiency. Therefore, MaxPTE optimizes $PTE$ under strong-connectivity constraints and bounded channel modifications, improving liquidity organization while preserving the fundamental reachability of the payment network.

\section{PCN Topology Optimization Algorithm}
\label{PCN Topology Optimization Algorithm}

Existing studies mainly mitigate balance deficits by adjusting routing decisions or reallocating pre-deposited balances according to observed or predicted transaction demands~\cite{khalil2017revive,avarikioti2022hide,hong2022cycle,ge2022shaduf,li2020secure,ni2023utility,wang2023fence}. In contrast, the role of channel establishment in shaping liquidity partitioning across the network has received considerably less attention. Under the channel-level balance-locking mechanism, different outgoing channels form independently constrained liquidity pools. Consequently, even when the total balance is sufficient, redundant or highly fragmented channel configurations may prevent the available liquidity from being efficiently utilized, thereby aggravating local balance shortages and disrupting network-wide balance circulation.

Motivated by this observation, we optimize the topology of the PCN rather than dynamically adapting channel balances to individual demand realizations. Specifically, we consolidate structurally redundant outgoing channels and reallocate their pre-deposited balances to the retained channels. This operation preserves the total balance owned by each node while changing how its liquidity is organized across the topology. The proposed $PTE$ metric is used to evaluate the resulting structural differentiation of node-level liquidity allocations and to select the channel consolidation that most improves the organization of the existing balance supply.

As illustrated by the toy PCNs in Table I, topologies with larger $PTE$
values tend to exhibit smaller balance deficits under identical balance
and demand settings. This observation indicates that topology
configurations with differentiated liquidity roles can provide more
efficient balance circulation. Therefore, $PTE$ serves as a structural
criterion for guiding topology reconfiguration toward more effective
liquidity organizations. However, a larger $PTE$ does not imply that arbitrary channel removal is beneficial. Although consolidating channels can reduce liquidity fragmentation and concentrate balances on more distinctive payment directions, removing too many channels may decrease path diversity, increase routing distance, or even damage network reachability.

To balance these competing effects, we constrain the topology optimization in two ways. First, every candidate topology must remain strongly connected, ensuring that each node can still transfer balance to every other node through at least one directed path. Second, we introduce an iteration parameter $k$ that limits the number of outgoing channels removed from each node during optimization. Therefore, $k$ controls the trade-off between increasing structural liquidity differentiation and preserving transaction-path availability. Based on these principles, we formulate the topology optimization objective as maximizing $PTE$ subject to bounded channel reduction and strong-connectivity constraints.

Formally, we propose a PCN topology optimization algorithm named MaxPTE,
which maximizes $PTE$ while preserving liquidity accessibility. Given an
initial PCN topology $G=(V,E)$, the optimization objective is formulated as

\begin{equation}
\begin{aligned}
\max_{G'=(V,E')}
\quad & PTE(G'),\\
\mathrm{s.t.}\quad
&0\leq d_i^{out}-d_i^{\prime out}\leq k,
\quad \forall i\in V,\\
& G' \text{ is strongly connected},\\
&\sum_j c'_{ij}=\sum_j c_{ij},
\quad \forall i\in V .
\end{aligned}
\label{eq:maxpte}
\end{equation}

where $d_i^{out}$ and $d_i^{\prime out}$ denote the outgoing degrees
before and after optimization, respectively. The second constraint
preserves network accessibility, while the third constraint guarantees
that topology reconfiguration only reorganizes existing liquidity
without introducing additional balance.

\begin{algorithm}
    \renewcommand{\algorithmicrequire}{\textbf{Input:}}
    \renewcommand{\algorithmicensure}{\textbf{Output:}}
    \caption{MaxPTE}
    \label{maxPTE}
    \begin{algorithmic}[1]
        \REQUIRE Initial PCN $G(V,E)$, total iteration round $k$
        \ENSURE Modified PCN $G'(V,E')$
        \STATE $G'(V,E') = G(V,E)$
        \FOR{$iter = 1$ to $k$}
            \FOR{each $v \in V$}
                \STATE $\textit{maxPTE} = \textit{PTE}(G')$, $\textit{tmpN} = \emptyset$, $\textit{tmpM} = \emptyset$
                \IF{$v$.number\_of\_channels $> 2$} 
                    % \COMMENT{The number of $v$'channels is not less than 2}
                    \FOR{$i ,j\in v.\text{neighbors}(i\neq j)$}
                        
                        \STATE $tmpG(V,tmpE) = G'(V,E')$
                        \STATE $tmpE$.balance($v, i$) += $tmpE$.balance($v, j$)
                        \STATE $tmpE$.remove$(v, j)$
                        \IF{$\textit{PTE}(tmpG) > \textit{maxPTE}$ and $tmpG$ is strongly connected}
                            \STATE $\textit{maxPTE} = \textit{PTE}(tmpG)$
                            \STATE $\textit{tmpN} = i$, $\textit{tmpM} = j$
                        \ENDIF
                        
                    \ENDFOR
                    \IF{$\textit{tmpN} \neq \emptyset$}
                        \STATE $E'$.balance($v$,\textit{tmpN}) += $E'$.balance($v$,\textit{tmpM})
                        \STATE $E'$.remove($v$, \textit{tmpM})
                    \ENDIF
                \ENDIF
            \ENDFOR
        \ENDFOR
    \end{algorithmic}
\end{algorithm}

The pseudocode is shown in \cref{maxPTE}.
Denote the optimized graph by $G'$. 
The algorithm begins by initializing $G'$ as the original graph $G$ (Line 1). 
It proceeds $k$ rounds iteratively (Line 2). 
In each round, each node $v$ in $V$ are sequentially processed by attempting to reduce one channel from itself (Line 3). 
Specifically, the $PTE$ of the current graph $G'$ is recorded in the variable $\textit{maxPTE}$, which tracks the highest $PTE$ achieved during the process. To facilitate the optimization, two temporary variables, $\textit{tmpN}$ and $\textit{tmpM}$, are introduced to store the current best merging scheme (Line 4). If the number of channels associated with node $v$ exceeds 2 (Line 5), the algorithm initiates the merging process for the corresponding channels.
Merging the channels $(v,tmpN)$ and $(v,tmpM)$ represents the current best scheme, which is determined through the following steps:

\textbf{\textit{Step 1:}} Identifying the optimal channel merging scheme, which maximizes the $PTE$ value, for node $v$ (Line 6-14): The two neighbors $i$ and $j$, of node $v$ are first identified (Line 6), and a temporary graph $\textit{tmpG}(V, E)$ is created to simulate the merging process. In $\textit{tmpG}$, the balances of channel $(v, i)$ are added to the channel $(v, j)$ (Line 7-8), and the channel $(v, j)$ is subsequently removed (Line 9). If the $PTE$ of $\textit{tmpG}$ exceeds the current $\textit{maxPTE}$ and $\textit{tmpG}$ remains strongly connected (i.e., each node can transfer balance to any other node in the network) (Line 10), the value of $\textit{maxPTE}$ is updated to reflect the new $PTE$ of $\textit{tmpG}$ (Line 11). Furthermore, the temporary variables $\textit{tmpN}$ and $\textit{tmpM}$ are updated to store the corresponding merging scheme (Line 12).

\textbf{\textit{Step 2:}} Updating the network based on the optimal merging scheme (Lines 15–18): If $\textit{tmpN}$ is non-null, it indicates the discovery of a merging scheme yielding a higher $PTE$ than the initial one for graph $G'$ at node $v$ (Line 15). The PCN is updated by adding the balances of channel $(v, \textit{tmpM})$ to channel $(v, \textit{tmpN})$ (Line 16) and subsequently removing channel $(v, \textit{tmpM})$ (Line 17). This ensures that network $G'$ is incrementally refined and optimized.

\begin{figure}[t]
    \centering
    \includegraphics[width=0.7\linewidth]{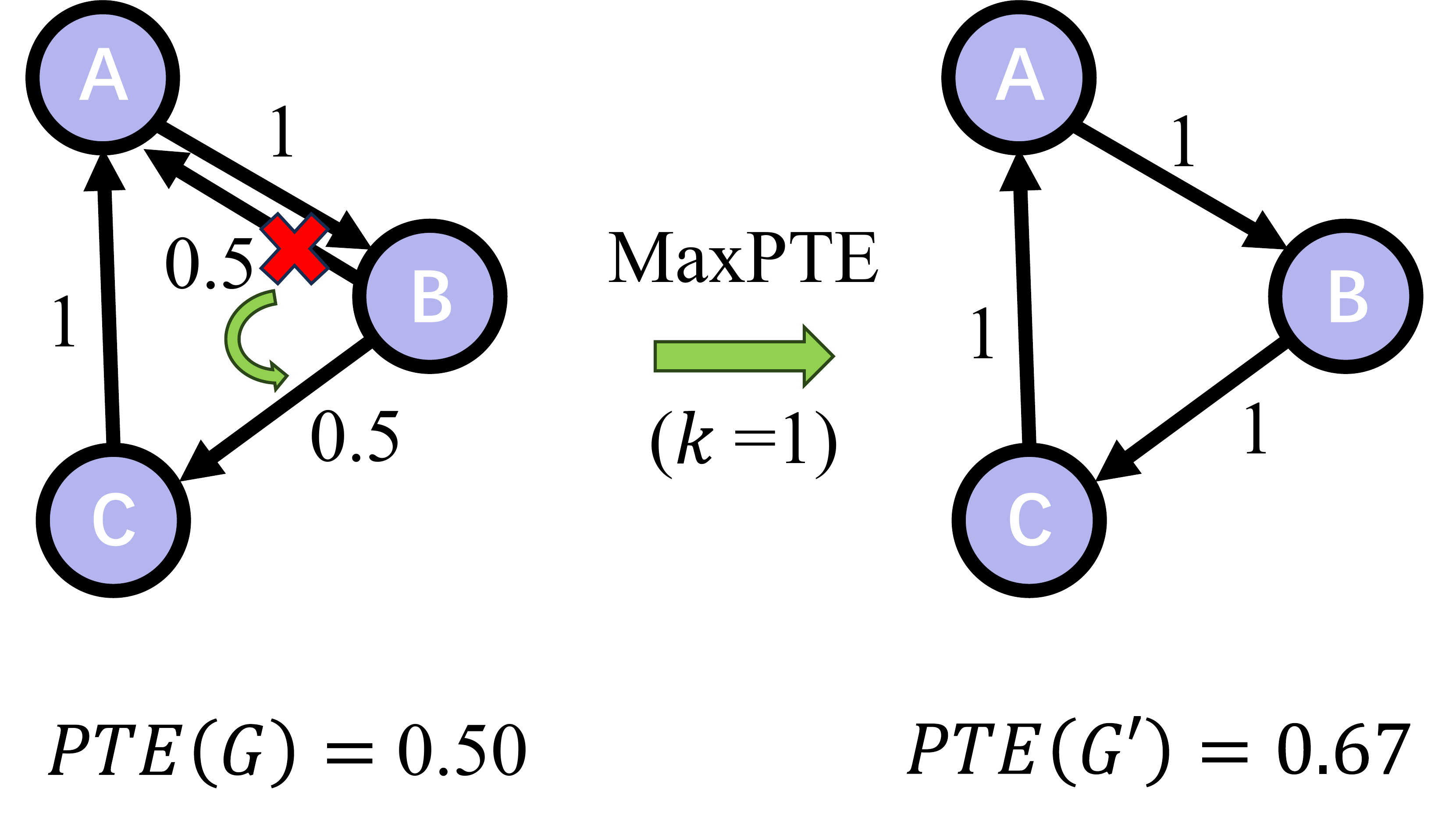}
    \caption{An example of the MaxPTE algorithm with one iteration (i.e., $k = 1$).  
    Before implementing the MaxPTE algorithm, $PTE(G) = $ 0.50. 
    By the MaxPTE algorithm, node B consolidates the channel towards node A into the channel towards node C, aggregating the corresponding balance while preserving the total liquidity supply.
    As a result, the graph $G'$ with the maximum $PTE$ is obtained, achieving $PTE(G') = $ 0.67.\looseness=-1}
    \label{maxPTEexample}
    \vspace{-0.2cm}
\end{figure}
 We use \cref{maxPTEexample} to provide an example of the MaxPTE  algorithm. It effectively ensures that each node undergoes at most $k$ channel merges, improving network structure with limited topology modifications. For changes in channel balance, when altering the topology, i.e., merging two channels, the balance is explicitly deployed into the merged channel.

\begin{theorem}
    The time complexity of MaxPTE is \(O(kn\Delta^2|E|)\), where
    $k$ is the maximum number of channel reductions for a single node,
    $\Delta$ denotes the maximum outgoing degree of nodes, and $|E|$
    represents the number of channels in the PCN.
\end{theorem}

\begin{proof}
The algorithm performs at most $k$ iterations, and each iteration
processes all $n$ nodes. For each node, MaxPTE enumerates pairs of
outgoing neighbors to identify the optimal channel consolidation scheme.
Since the number of neighbors is bounded by the maximum outgoing degree
$\Delta$, the number of candidate consolidations for each node is
$O(\Delta^2)$. For each candidate consolidation, the algorithm evaluates
the resulting $PTE$. 
Computing $PTE$ requires constructing the balance allocation
distributions and calculating the deviations between node-level and
network-level liquidity patterns, which takes $O(|E|)$ time using the
adjacency representation of the PCN.
Therefore, considering all nodes and iterations, the overall time complexity is \(O(kn\Delta^2|E|)\).
\end{proof}

Although MaxPTE has polynomial complexity, it is designed as a topology maintenance mechanism rather than an online routing algorithm, and is therefore executed periodically during channel establishment or topology adjustment instead of for individual transactions. Moreover, real-world PCNs usually exhibit sparse connectivity. The average node degree is approximately 7.43~\cite{lightningStats2023}; the LTC mainnet contains only 95 nodes, and the global Lightning Network contains 14,872 nodes~\cite{lightningExplorer2025}. Therefore, MaxPTE is mainly affected by local neighborhood structures rather than the global network scale, making it applicable to practical PCN environments.
\section{Experimental Evaluation}
\label{sec:Experimental Evaluation}

In this section, we adopt three representative metrics, consisting of balance deficit, average maximum flow, and transaction failure probability, which have been widely used to assess liquidity efficiency and network performance in prior PCN studies \cite{rohrer2019discharged,khalil2017revive,sivaraman2020high,liu2024balanced}. These studies focus on balance management or routing rather than topology modification, and are cited only for their metric formulations. We conduct comprehensive evaluations of performance in PCN on a 64-bit Ubuntu 16.04 operating system utilizing the Lightning Network Daemon \cite{LightningNetworkDaemon} and Python's NetworkX \cite{hagberg2008exploring}. 

\subsection{Necessity Analysis}

We first demonstrate the necessity of consolidating fragmented channels to improve liquidity utilization. Consider a PCN containing multiple independent channels, where each channel reserves a portion of the total available balance. Since channel balances are independently locked, these separated liquidity pools cannot be directly shared when transaction demands arrive. Therefore, even when the aggregate balance is sufficient, fragmented channel configurations may lead to local balance shortages.

As illustrated in \cref{independent-combined-fig}(a), the original PCN contains $n$ independent channels, and the balance deployed in each channel is denoted by $\mu$. When these fragmented channels are consolidated, their balances are aggregated into a combined channel with total capacity $n\times\mu$, as shown in \cref{independent-combined-fig}(b). The total liquidity supply remains unchanged, while the available balance becomes less fragmented and can serve a wider range of transaction demands.

\begin{figure}[ht]
    \centering
    \includegraphics[width=0.75\linewidth]{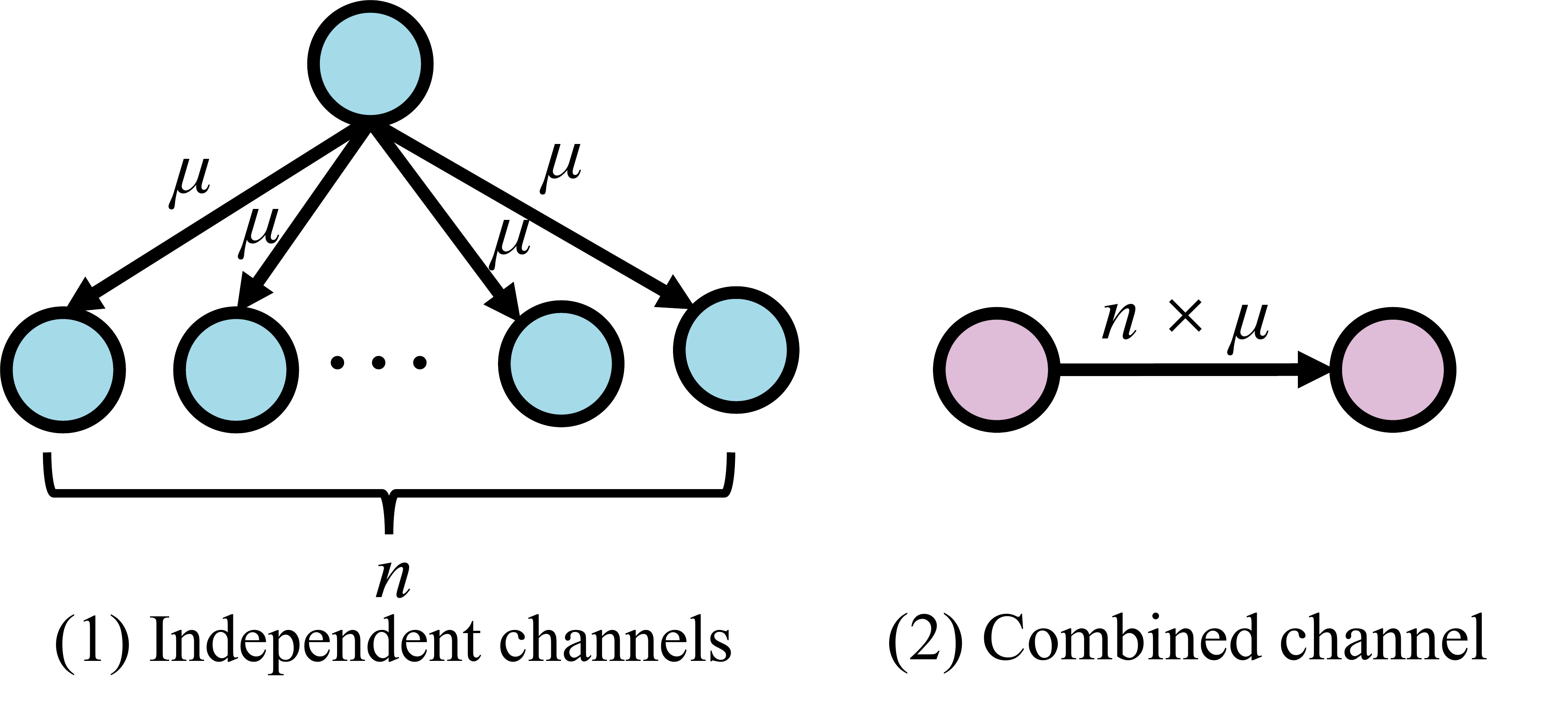}
    \caption{Schematic diagram of fragmented channels before consolidation and the aggregated channel after consolidation.}
    \label{independent-combined-fig}
\end{figure}

\begin{figure}[t]
    \centering
    \includegraphics[width=1\linewidth]{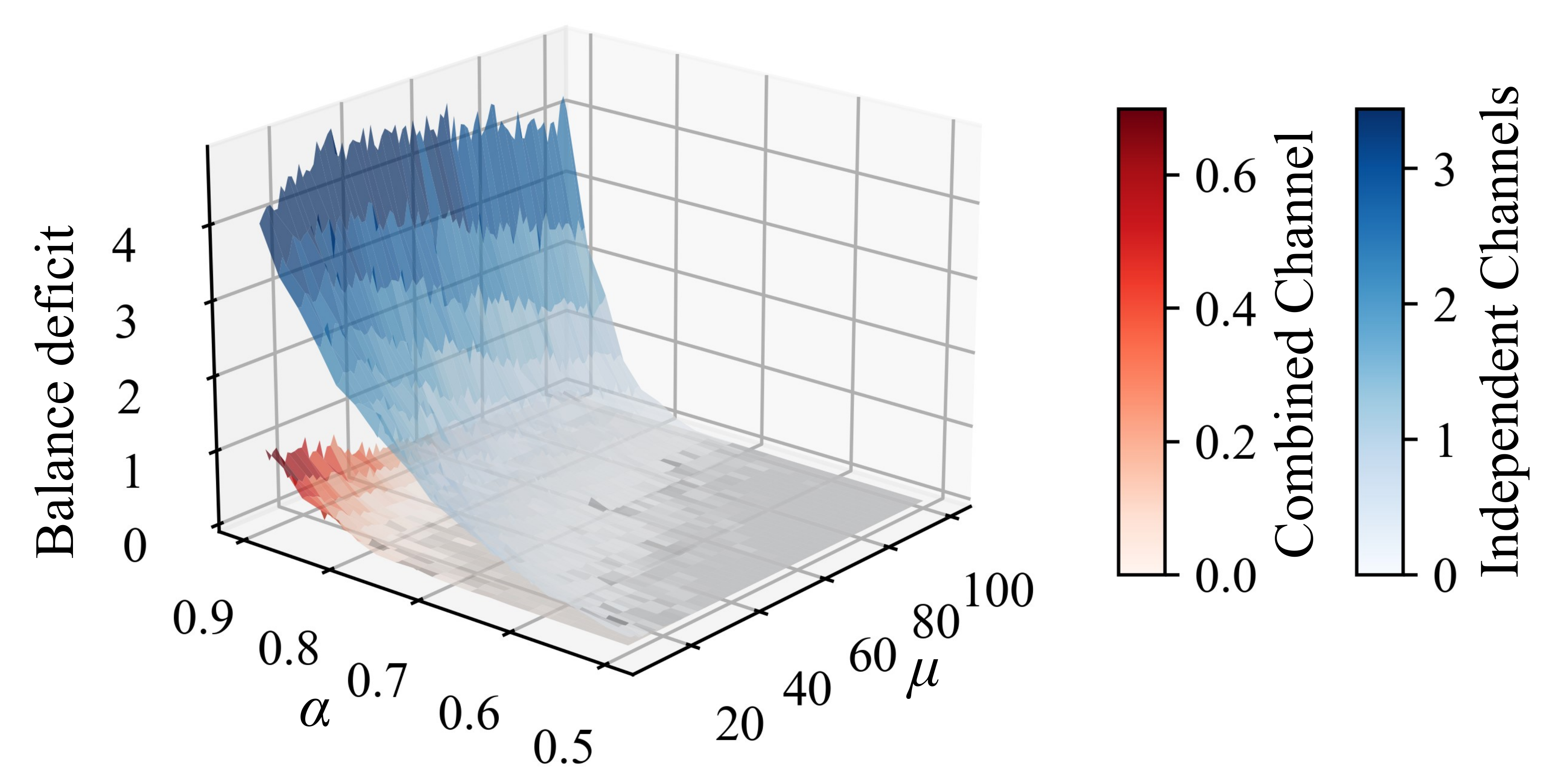}
    \caption{Balance deficit comparison.}
    \label{plane}
    \vspace{-0.2cm}
\end{figure}

To evaluate the occurrence of supply-demand imbalance risk within channels under varying degrees of supply-demand urgency, we introduce the load factor {$\alpha$} $\in[0.5, 1]$.  Under the conditions $\mu\in[0, 100]$ ( The definition of $\alpha$ and $\mu$ are illustrated in \cref{2pcnCompare}) and ${n = 10}$, we calculate the balance deficit generated by both the combined channel and the independent channels. To mitigate the impact of extreme probabilities, we conducted 1,000 independent experiments and used the average values as the balance deficit for each channel type.
As shown in \cref{plane}, the blue plane represents the balance deficit of the independent channels, while the red plane represents that of the combined channel. As $\alpha$ increases, the balance deficit in the channels grows. Additionally, regardless of changes in $\alpha$ and $\mu$, the balance deficit from the combined channel is consistently lower than that of the independent channels. Consolidating fragmented channels improves the efficiency of existing balance supply and increases the likelihood of meeting transaction demands.

\subsection{Balance deficit}
\label{subsec:balance_deficit}

This subsection evaluates how different channel allocation strategies respond to heterogeneous transaction demands through four representative distributions. Using 1ML’s channels-per-node percentiles~\cite{lightningStats2023} (e.g., 25\% of nodes have $\ge\!4$ channels, 5\% $\ge\!25$, 1\% $\ge\!96$), we generate 10 Lightning Network topologies (200 nodes and about 750 edges each). Each node allocates a total balance of 100 satoshis to its payment channels. We then apply our proposed MaxPTE-based optimization and compare the results with networks constructed by several baseline channel-consolidation algorithms. Specifically, we consider the following methods:\\
  \textbf{(1) Random:} Each node randomly merges two channels.\\
  \textbf{(2) MaxOut:} Each node merges channels with its two neighboring nodes having the highest degree.\\
  \textbf{(3) MinOut:} Each node merges channels with its two neighboring nodes having the lowest degree.\\
 \textbf{(4) MaxBetweenness:} Each node selects its two channels with the highest betweenness~\cite{freeman1977set} for merging. The betweenness of an edge is the ratio of the number of shortest paths that pass through it to the total number of shortest paths in the network.\\
 \textbf{(5) MaxClustering:} Each node selects two of its channels for merging to maximize the network’s clustering coefficient~\cite{watts1998collective}. The clustering coefficient of a node is defined as the ratio of the number of existing edges among its neighbors to the maximum possible number of such edges, and the network clustering coefficient is the average across all nodes.
 
For fairness, all networks preserve the same total number of channels as the MaxPTE-optimized topology, and we fix $k\!=\!1$ for MaxPTE in this evaluation.

The metrics we used are as follows. According to the definition of maximum flow in a network~\cite{Ford_Fulkerson_1956}, for any two nodes $i$ and $j$ ($i\!\neq\!j$), the maximum transferable balance through all paths is defined as
\(
F_{ij} = \max \left( \sum_{p \in P_{ij}} f_p \right),
\)
where $P_{ij}$ denotes the set of all valid paths between nodes $i$ and $j$, and $f_p$ is the transferable balance along path $p$. The total maximum flow of a PCN can be written as
\(
MF = \sum_{i \neq j} F_{ij}.\label{MF}
\)
The balance deficit between the node pair is defined as
\(
\text{Deficit}_{ij} = \max\{\,0,\; D'_{ij} - F_{ij}\,\},
\)
where $D'_{ij}$ represents the balance demand from node $i$ to node $j$. The overall balance deficit is obtained by averaging $\text{Deficit}_{ij}$ across all node pairs.

To assess the robustness of our approach under different transaction patterns, we evaluate balance deficits under four representative demand distributions~\cite{ross2014introduction}. The variable $Demand$ serves as a global scaling factor to uniformly amplify the generated demand matrix and simulate different traffic intensities.

\textbf{(a) Poisson:} Transaction demands between node pairs are independently generated according to a Poisson distribution $M_{ij}\sim\mathrm{Poisson}(\lambda)$ with $\lambda=20$, reflecting spontaneous and uncorrelated payment requests in steady-state PCNs.

\textbf{(b) Uniform:} Transaction demands are uniformly distributed across node pairs as $M_{ij}\sim U(0,2\lambda)$ with $\lambda=20$, simulating balanced liquidity utilization across the topology.

\textbf{(c) Power-law:} To model hub-dominated transaction patterns, node weights are generated according to their connectivity ranks as $w_i\propto1/\mathrm{rank}(i)^s$ with $s=0.5$. The final weights combine power-law and uniform components as $w=\alpha w_{\mathrm{power-law}}+(1-\alpha)w_{\mathrm{uniform}}$ with $\alpha=0.5$, and the demand matrix is constructed as $M=2\lambda ww^{\top}$.

\textbf{(d) Gaussian:} To capture locality-aware transaction patterns, demands are concentrated around three highly connected nodes ($|C|=3$), and the pairwise demand is generated according to $M_{ij}=\sum_{c\in C}\exp(-(d(i,c)^2+d(j,c)^2)/(2\sigma^2))$, where $\sigma=2.0$. This captures transactions concentrated around structurally central nodes such as merchants or liquidity hubs.

\begin{figure}[t]
    \centering
    \includegraphics[width=1\linewidth]{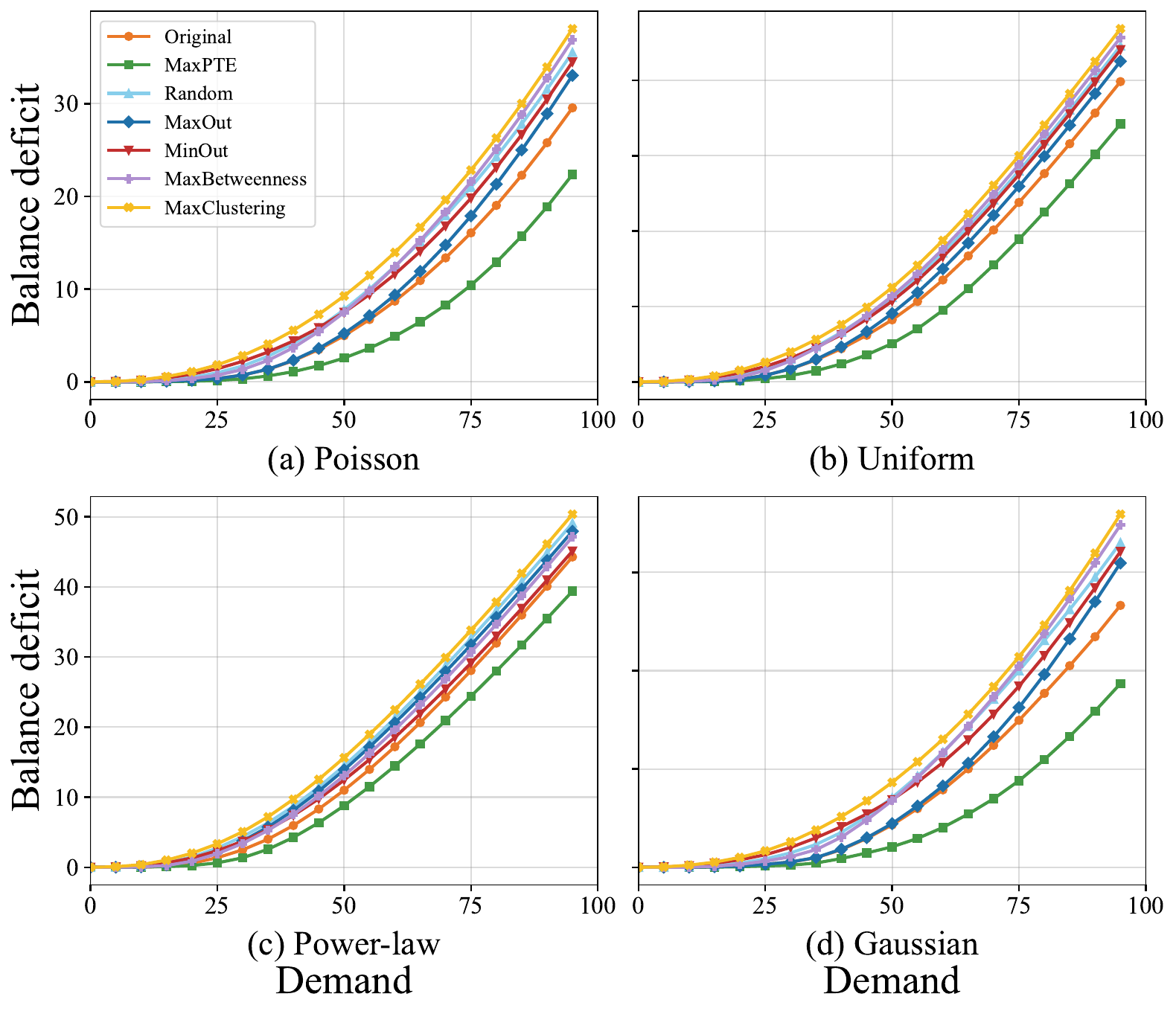}
    \caption{Average balance deficits under different demand distributions.}
    \label{4demands}
    \vspace{-0.2cm}
\end{figure}

As we can see, \cref{4demands} shows the average balance deficits under different demand distributions. As transaction demand increases, balance deficits grow accordingly. Notably, our MaxPTE-optimized PCN consistently yields lower balance deficits than both the original and other baseline networks. Across all demand models, our MaxPTE algorithm reduces the average balance deficits by 27.25\% compared to the original network and is considerably lower than that of other algorithms, indicating that it effectively reduces balance deficits and mitigates supply-demand imbalances in PCNs, outperforming other strategies. The improvement is especially significant under heterogeneous demand patterns (power-law and Gaussian distributions), where balance allocation and channel connectivity are highly uneven. This confirms that the MaxPTE algorithm enhances balance robustness by mitigating supply–demand mismatches and demonstrates robustness across different demand distributions.

\subsection{Average maximum flow}

Based on the above definition of the maximum flow ($MF$) between nodes, we define the average maximum flow $(AMF)$ of the network as the average of the maximum flows between all node pairs in PCN, which can be expressed as:
\[
AMF = \frac{1}{n(n-1)} \sum_{i \neq j} \max \left( \sum_{p \in P_{ij}} f_p \right),
\label{AMF}
\]
where $n$ is the total number of nodes in the network, and ${n(n-1)}$ represents the number of node pairs with direction. \looseness=-1

\cref{box} uses box plots to present the PCN average maximum flow values under different channel merging algorithms. The red vertical line in the center of each box plot represents the median of multiple experimental results. It can be observed that the MaxPTE algorithm improves the average maximum flow by 12.23\% compared to the original network. Moreover, MaxPTE algorithm significantly outperforms other channel merging strategies.  Furthermore, with the analysis of balance deficit above, \cref{box} also explains the reasons behind the superior performance of the MaxPTE algorithm --- the increase in maximum flow enhances the mobility and circulation of balance within the network, significantly contributing to mitigating supply-demand imbalances.

\begin{figure}[t]
    \centering
    \includegraphics[width=1\linewidth]{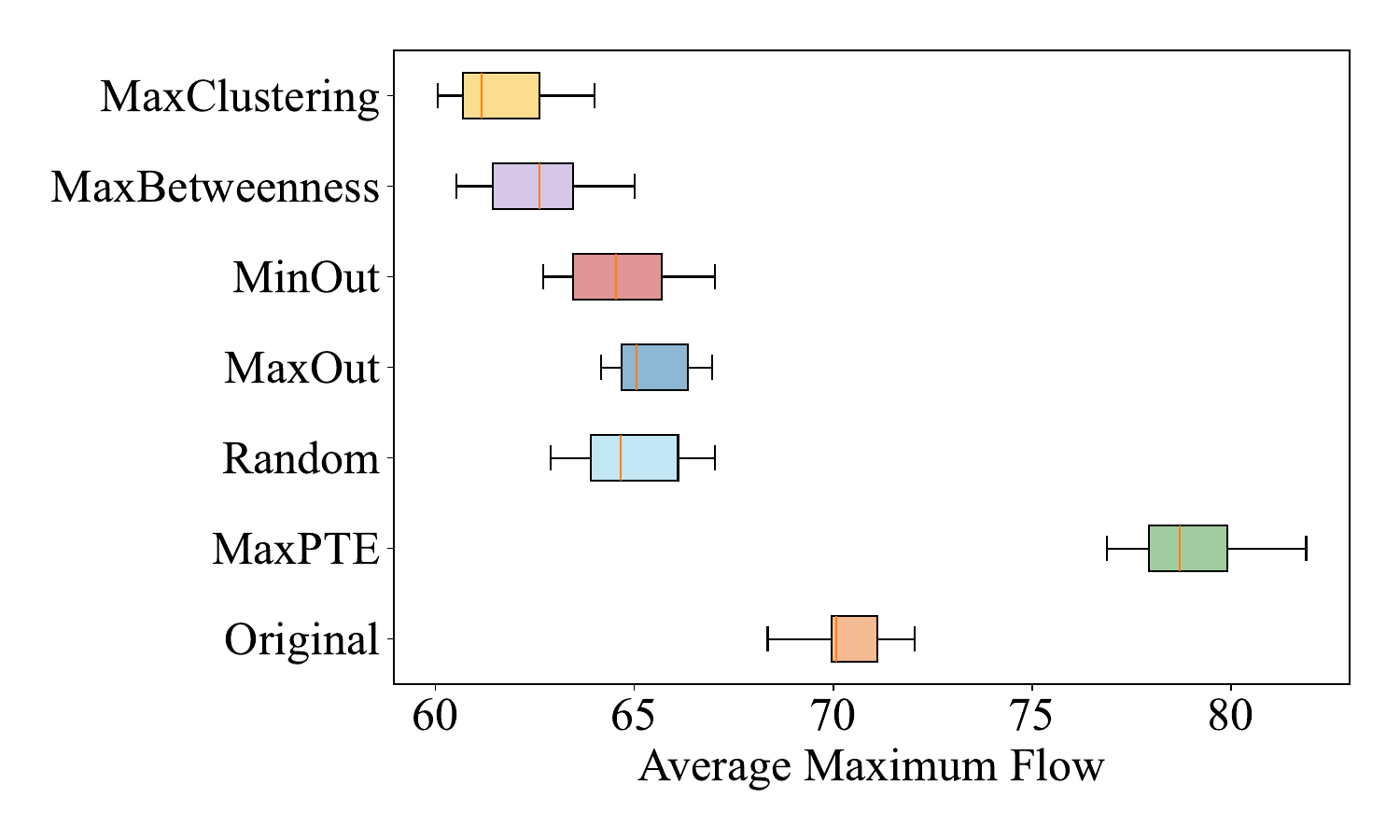}
    \caption{Average maximum flows of different PCNs.}
    \label{box}
    \vspace{-0.2cm}
\end{figure}

\subsection{Proportion of failure transactions}

\begin{figure}[t]
    \centering
    \includegraphics[width=1\linewidth]{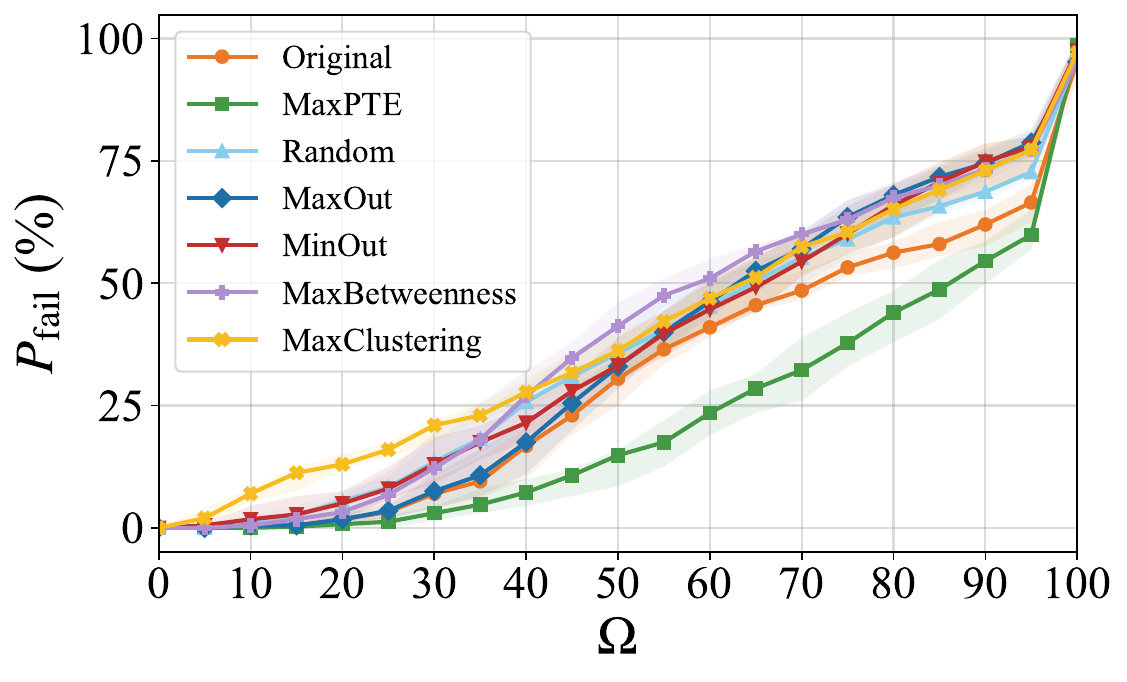}
    \caption{Proportion of failure transactions.}
    \label{badpairs}
    \vspace{-0.2cm}
\end{figure}

We define $P_{\mathrm{fail}}(\Omega)$ as the fraction of transaction requests that cannot be successful. 
We refer to $\Omega$ as the transaction balance demand, indicating the required amount for network routing feasibility.
Intuitively, higher $\Omega$ values correspond to larger balance requirements from the network.

Formally, let $\mathcal{T}$ be the set of all node pairs, and let $F_{\max}(s,d)$ denote the maximum available flow between source $s$ and destination $d$ under current channel balances. 
A transaction with balance demand $\Omega$ from $s$ to $d$ is deemed to \emph{fail} if $F_{\max}(s,d) < \Omega$. The failure probability metric follows prior approaches that quantify reliability and flow failures in PCNs \cite{rohrer2019discharged}.
Then we have
\begin{displaymath}
P_{\mathrm{fail}}(\Omega)
= \frac{1}{|\mathcal{T}|}
\sum_{(s,d)\in\mathcal{T}} \mathbf{1}\!\left[\,F_{\max}(s,d) < \Omega\,\right],
\end{displaymath}
which measures the share of transactions that cannot be fulfilled at the balance demand level $\Omega$. 
A higher $P_{\mathrm{fail}}(\Omega)$ indicates that fewer node pairs can forward transactions requiring at least $\Omega$ satoshis in total, implying a stronger supply–demand imbalance in the network.

As shown in \cref{badpairs}, the shaded bands represent the variation across 10 different network topologies. It is evident that MaxPTE algorithm achieves the lowest average failure probability, outperforming all baseline methods.
MaxPTE algorithm reduces the mean transaction failure rate by 25.96\% compared to the original topology. 
These quantitative results confirm that MaxPTE achieves the best performance among evaluated PCNs, significantly reducing transaction failures even under high balance demands, while alternative heuristics experience substantially higher $P_{fail}(\Omega)$ values on average.

\subsection{Discussion}
The evaluation results demonstrate that MaxPTE effectively improves
liquidity utilization and reduces transaction failures by reorganizing
the structural allocation of existing balances. Different from
demand-driven balance adaptation methods, MaxPTE provides a topology-level optimization perspective, enabling PCNs to enhance supply-demand
alignment through structural reconfiguration. This property makes
MaxPTE complementary to existing balance management and routing
optimization approaches. 

\section{Related Work}
\label{sec:Related Work}
To enhance the transaction success rate in PCN while maintaining supply-demand balance, related solutions can be generally categorized into transaction path optimization and channel balance management.

The methods of transaction routing optimization are illustrated by typical works: Spider \cite{sivaraman2020high}, which employs a multi-path partitioned transmission technique that divides transactions into smaller packets transmitted along separate paths. This method reduces strain on paths and enables better resource use and dynamic traffic management. FSTR \cite{lin2020fstr} introduces a path selection mechanism based on balance skewness, prioritizing routes with sufficient, evenly distributed balances to minimize path exhaustion risk and ensure a strategic selection process. Boomerang \cite{bagaria2020boomerang} enhances transaction success by incorporating redundant payment routes. During execution, payment packets are dispatched simultaneously across multiple paths. Once the recipient confirms receipt of the full amount, any unnecessary routes are immediately terminated to avoid excessive payments. Spear \cite{rahimpour2021spear} employs a similar method, ensuring at least one route successfully delivers the necessary balance. Upon payment confirmation, redundant paths are quickly interrupted, enhancing efficiency. Robust Pay+ \cite{zhang2021robustpay+} tackles the challenge of reliable routing by proposing a distributed algorithm that identifies the minimum-cost path in adverse conditions, maintaining stability and managing expenses effectively. BRBW \cite{liu2024balanced} combines an enhanced maximum flow algorithm with the Analytic Hierarchy Process (AHP) to locate routes with sufficient capacity while constructing a weight model that considers multiple factors, creating a sophisticated framework for optimizing transaction paths regarding both performance and reliability.\looseness=-1

Additionally, balance planning schemes for payment channels primarily include the following approaches: Revive \cite{khalil2017revive} enables users in PCN to securely rebalance channel balances according to channel owners' preferences without closing or topping up channels, ensuring flexible balance management. It also achieves global rebalancing through localized adjustments, reducing the risk of balance shortages across the network. HIDE \& SEEK \cite{avarikioti2022hide} proposes a participatory rebalancing protocol that utilizes multi-party computation to solve linear programming problems, preserving privacy while achieving optimal balance rebalancing and preventing disclosure of sensitive balance distributions among participants. Cycle \cite{hong2022cycle} maintains balance equilibrium within PCN through asynchronous rebalancing and circular payments, maximizing balance utilization across interconnected channels and enhancing rebalancing efficiency. Shaduf \cite{ge2022shaduf}, a novel off-chain rebalancing protocol, optimizes balance management in payment channel networks by transferring balances directly instead of relying on periodic configurations, thereby reducing dependence on liquidity cycles. PnP \cite{li2020secure} utilizes approximation algorithms and encrypted sorting to determine the deployment strategy for initial balance allocation in payment channels, ensuring that balance distribution meets transaction demands while optimizing balance utilization during the channel's initialization phase. UAR \cite{ni2023utility} optimizes payment channel liquidity by shifting tokens from high-balance to low-balance and high-utility channels, maximizing transaction success rates. The Fence \cite{wang2023fence}, utilizing an online competitive algorithm, dynamically allocates balance through real-time adjustments, adapting to network changes while minimizing inefficiencies and delays in the planning process.

However, the effectiveness of the aforementioned methods largely depends on demand prediction accuracy. Due to dynamic changes in transaction demands and channel conditions, these methods often require frequent updates to adapt to evolving balance allocation and network states. 
Our work differs from existing studies by optimizing the structural organization of liquidity through topology reconfiguration rather than adjusting routing decisions or balance allocations.
 
\section{Conclusion}
\label{sec:Conclusion}
Although PCNs provide an effective solution to blockchain scalability, balance shortages remain a critical challenge. This paper reveals the impact of topology structure on balance efficiency and proposes $PTE$ to characterize structural liquidity patterns. Based on $PTE$, we develop MaxPTE to optimize channel organization while preserving network accessibility. Extensive experiments under various balance-demand distributions show that MaxPTE reduces balance deficits by 27.25\%, improves average maximum flow by 12.23\%, and decreases transaction failure probability by 25.96\% compared with baseline methods, demonstrating its effectiveness without relying on prior demand prediction.

\bibliographystyle{IEEEtran}
\bibliography{reference}

\end{document}